\documentclass[a4paper,12pt]{article}
\usepackage{complexity}
\usepackage{lineno}
\usepackage{xcolor}
\usepackage[T1]{fontenc}
\usepackage{graphicx}
\usepackage{textgreek}
\usepackage{pstricks,pspicture}
\usepackage{float}
\usepackage{amsthm,amsmath,amssymb}
\usepackage{mathrsfs,amssymb,amsmath,amsthm,amstext,amsfonts}
\usepackage{pst-plot}
\usepackage{epsfig}
\usepackage{setspace}
\usepackage{float}
\usepackage{rotating}
\usepackage{tikz}
\usepackage[shortlabels]{enumitem}
\usepackage[left=2cm,top=2cm, bottom=2cm,right=2cm]{geometry}

\newcommand{\intv}{I_{cc}}
\DeclareMathOperator{\hullc}{Hull_{cc}}

\DeclareMathOperator{\hn}{hn_{cc}}

\DeclareMathOperator{\hullgt}{Hull_{\cal C}}
\DeclareMathOperator{\rk}{rk_{cc}}
\DeclareMathOperator{\hel}{hl_{cc}}
\DeclareMathOperator{\rad}{rd_{cc}}

\newtheorem{theorem}{Theorem}

\newtheorem{conjecture}{Conjecture}

\title{Computing the Helly Number, Radon Number and Rank in Cycle Convexity}
\author{
\shortstack{
Revathy S. Nair$^{a}$\footnote{revathyrahulnivi@gmail.com}
\quad
Bijo S. Anand$^{b}$\footnote{bijos\_anand@yahoo.com}
\quad
Ullas Chandran S. V.$^{c}$\footnote{svuc.math@gmail.com}
\\[2mm]
Julliano R. Nascimento$^{d}$\footnote{jullianonascimento@ufg.br}
\quad
Arun Anil$^{b}$\footnote{arunanil93@gmail.com}
}
\\[3mm]
$^{a}$\small Department of Mathematics, Mar Ivanios College, University of Kerala,
\small Thiruvananthapuram, India
\\
$^{b}$\small Department of Mathematics, Sree Narayana College, Punalur, Kerala
\\
$^{c}$\small Department of Mathematics, Mahatma Gandhi College,
\small Thiruvananthapuram, Kerala, India
\\
$^{d}$\small Instituto de Informática, Universidade Federal de Goiás,
\small Goiânia, GO, Brazil
}

\begin{document}

\maketitle

\begin{abstract}
In this paper, we investigate three fundamental convexity parameters of graphs under cycle convexity, namely the Helly number, Radon number, and rank. We first study the computational complexity of these parameters. For each of these parameters, we consider the associated threshold decision problem of determining whether the parameter of a given graph is at least a prescribed integer. We establish that all three problems are $\NP$-hard and $\W[1]$-hard when parameterized by the threshold. Moreover, we strengthen these results by showing that the $\NP$-hardness persists even when the input is restricted to planar graphs of maximum degree at most $6$. We also focus on the structural properties of connected graphs corresponding to extremal values of these parameters. In particular, we characterize the graph classes for which the three parameters attain the values $n-1$ and $n-2$, where $n$ is the order of $G$.
 \end{abstract}

    {\bf keywords}Convexity, Helly number, Radon number, Rank\\
    \textbf{AMS Sub. Classification:} {05C69, 05C76, 05C85}

 \section{Introduction}\label{section_introduction}
 The study of convexity spaces in mathematics form an active and productive area of research, which provides a framework for the generalized concepts that go beyond the classical notions of convex sets in the Euclidean geometry. Graph convexities started gaining attention in the literature since the 1970s, beginning with a paper by Erd{\H{o}}s et al~\cite{erdHos1972some} and subsequently got further interest through the contributed works by Harary and Nieminen~\cite{harary1981convexity}, Mulder~\cite{mulder1980interval}, Farber and Jamison~\cite{farber1986convexity}, and others.

In general, for a graph $G$, a set ${\cal C}$ of subsets of $V(G)$ is a \textit{convexity} in $G$ if $(i)$ $\emptyset, V(G) \in {\cal C}$ and $(ii)$ ${\cal C}$ is closed under intersection. Each element of ${\cal C}$ is called a \textit{convex set}. The \textit{convex hull}, $\hullgt (S) $ of a subset $S \subseteq V(G)$ is the smallest convex set containing $S$. For more details, we refer to~\cite{duchet1987convexity,van22}. One of the most common way of defining a graph convexity in a given graph $G$ is through the family of paths $\mathcal{P}$ in $G$. A set $S \subseteq V(G)$ is $\mathcal{P}$-convex, if it contains all vertices of every path in $\mathcal{P}$ between any two vertices of $S$. Some classical example includes \textit{geodesic convexity}, where $\mathcal{P}$ is the set of all shortest paths in $G$~\cite{everett1985hull,buckley-1990,farber-1986}, \textit{monophonic convexity}~\cite{caceres-2005,source17,duchet1987convexity} and $P_3$ convexity~\cite{source11,centeno2011irreversible,coelho2019p3}, defined over induced paths and paths on three vertices, respectively.

Apart from this, there are  convexity notions that are not defined in terms of the path systems. Some example includes \emph{Steiner convexity}~\cite{source9} and $\Delta$-\emph{convexity}~\cite{bijo2,bijo3,bijo1,anand2025helly,anand2025carath}. For Steiner convexity, a set $S \subseteq V(G)$ is \emph{Steiner convex} if, for any subset $S' \subseteq S$, all vertices in every Steiner tree of $S'$ belong to $S$. In $\Delta$-convexity, a set $S\subseteq V(G)$ is $\Delta$-\emph{convex} if every vertex $u\in V(G)\setminus S$ cannot form a triangle with any two vertices in $S$. Graph convexities have been widely
studied across a range of settings, with special focus on the convexity invariants like the hull number, the interval number, and the convexity number. For an in-depth analysis of geodesic convexity, one can refer the extensive work of Pelayo~\cite{pelayo-2015}. Furthermore, a recent book \cite{araujo2025introduction} offers a comprehensive survey on the computational aspects of graph convexity. Moreover, graph convexities provide a natural framework for modelling the transmission processes, which includes the spread of infections, opinions, or information~\cite{dreyer2009irreversible}.

A recently studied graph convexity, called \textit{cycle convexity}, was introduced by Araújo et al.~\cite{interval08} and is motivated by applications and questions arising in graph theory and related areas, particularly knot theory~\cite{araujo2020cycle}. Several classical graph parameters have been studied in the context of cycle convexity. The interval number and hull number were introduced in~\cite{interval08} and~\cite{hull09,anand2025complexity}, respectively, while the convexity number was studied in~\cite{anand2025complexity,lima2024complexity}. More recently, other parameters, including the rank and general position numbers~\cite{araujo2025on}, convex partitions~\cite{gomes2025some}, percolation time~\cite{lima2024complexity}, exchange number~\cite{nair2026computing}, and Carath\'{e}odory number~\cite{nair2026carath}, have been investigated. 

In this paper, we focus on the classical parameters Helly number, Radon number, and rank in the context of cycle convexity. Although the rank in cycle convexity has been recently studied by Araújo et al.~\cite{araujo2026rank}, there is still room for further investigation, particularly under structural restrictions. We show that the threshold problems associated with all three parameters are $\NP$-hard even for planar graphs of maximum degree at most $6$. We also discuss their parameterized complexity with respect to the threshold. Furthermore, we characterize the connected graphs for which these parameters attain the values $n-1$ and $n-2$.

\section{Preliminaries}
All the graphs considered in this paper are connected, simple, and undirected, denoted by $G=(V,E)$.
Given a graph $G$ and $u \in V(G)$, the \index{open neighbourhood} open and the  \index{closed neighbourhood} closed neighbourhood of a vertex $u$ are $N_{G}(u) = \{v : uv \in E(G)\}$ and $N_{G}[u] = N_{G}(u) \cup \{u\}$, respectively. A vertex $v$ in a connected graph $G$ is a \index{cut-vertex} \textit{cut-vertex} of $G$ if $G - v$ is disconnected. A \textit{component} of a given undirected graph $G$ may be defined as a connected subgraph that is not part of any larger connected subgraph. Given a graph $G=(V,E)$ and a vertex set $S \subseteq V$, we denote by $G[S]$ the \emph{subgraph of $G$ induced by $S$},
that is, the graph with vertex set $S$ and edge set
$E(G[S]) = \{ uv \in E(G) : u,v \in S \}$. A path graph, denoted by $P_n$ is a graph with vertex set $\{v_1,v_2,\ldots,v_n\}$ and edges $\{v_i,v_{i+1}\}$ for $i\in \{1,2,\ldots,n-1\}$. A graph $G$ is \textit{unicyclic}, if it is connected and contains exactly one cycle. A \textit{bicyclic graph} is a simple connected graph with $n$ vertices and $n+1$ edges. A graph which can be drawn in the plane in such a way that edges meet only at points corresponding to their common ends is called a \textit{planar graph}, and such a drawing is called a \textit{planar embedding} of the graph. A subset of vertices in a graph $G$ is \textit{independent} if the vertices are pairwise non-adjacent. The \textit{independence number} of a graph $G$, denoted by $ \alpha (G)$, is the cardinality of the maximum independent set of vertices. 

For a set $S$ of vertices in a graph $G$, the \textit{cycle interval} of $S$, denoted by $\intv(S)$, is the set formed by the vertices of $S$ and any $w \in V(G)$ that form a cycle with some vertices of $S$. We say that a vertex $w \in V(G)\setminus S $ is \textit{generated} by $S$ if $w\in \intv(S)$. If $\intv(S) = S$, then $S$ is \textit{cycle convex} in $G$. The \textit{cycle convex hull} of a set $S$, denoted by $\hullc (S) $, is the smallest cycle convex set containing $S$. The cycle convex hull of $S$ can be obtained by successive applications of the cycle interval operation. For that, we define $\intv^0(S) = S$, $\intv^1(S) = \intv(S)$, and $\intv^k(S) = \intv(\intv^{k-1}(S))$, for every $k \geq 2$. With this terminology, $\hullc (S) = \intv^k(S)$, for any $k \geq 1$ such that the equality $\intv^k(S) = \intv^{k-1}(S)$ holds. If $\hullc(S) = V(G)$, then $S$ is a \textit{cycle hull set} of $G$. The $\textit{ cycle hull number}$ of $G$, denoted by $\hn(G)$, is the cardinality of a smallest set $S$ such that $\hullc(S) = V(G)$. We formally define the following three classical invariants with respect to cycle convexity.
\\
\textbf{Helly number:} A set $S$ is called \emph{Helly independent} (or \emph{H-independent})\index{Helly independent} if $\bigcap_{a\in S}\hullc(S\setminus\{a\})=\emptyset$,
and \emph{Helly dependent} (or \emph{H-dependent})\index{Helly dependent} otherwise. The \textit{Helly number}\index{Helly number} of $G$ with respect to cycle convexity, denoted by $\hel(G)$, is the maximum cardinality of a Helly independent set of $G$. 
\\
\textbf{Radon number:} A set $S$ is called \emph{Radon dependent} (or \emph{R-dependent}) if there exists a partition ${S_1,S_2}$ of $S$ such that $\hullc(S_1)\cap\hullc(S_2)\neq\emptyset$. In this case, ${S_1,S_2}$ is called a \emph{Radon partition} of $S$. If no such partition exists, then $S$ is called \emph{Radon independent} (or \emph{R-independent}). The \textit{Radon number} of $G$, denoted by $\rad(G)$, is the maximum cardinality of a Radon-independent set of $G$. 
\\
\textbf{Rank:} A set $S$ of vertices is called \emph{convexly independent} if $a\notin\hullc(S\setminus\{a\})$ for every $a\in S$. Otherwise, $S$ is called \emph{convexly dependent}. The \textit{rank} of $G$, denoted by $\rk(G)$, is the maximum cardinality of a convexly independent set of $G$.

We record the following useful observations concerning these three parameters. Let $\pi(G)$ denote any one of $\hel(G)$, $\rad(G)$, or $\rk(G)$,  and let $\Pi$ denote the corresponding notion of independence. Then the following basic observations hold. For every graph $G$ of order $n(G)\geq3$, we have $2\leq\pi(G)\leq n(G)$ and $\pi(G)\geq\alpha(G)$. Moreover, if every pair of adjacent vertices of $G$ is a hull set of $G$, then $\pi(G)=\max\{2,\alpha(G)\}$. Also, $\pi(G)=n(G)$ if and only if $G$ is a forest. Finally, note that, for any $\Pi$-independent set $S$ of a graph $G$, the induced subgraph $G[S]$ is a forest.

\section{Computational Complexity}

In this section, we investigate the computational complexity of determining the Helly number, the Radon number, and the rank in cycle convexity. The computational complexity of the rank has recently been studied by Araújo et al.~\cite{araujo2026rank}. In particular, for a graph $G$ and a positive integer $k$, they proved that deciding whether $\rk(G)\geq k$ is $\NP$-complete and $\W[1]$-hard when parameterized by $k$. They also observed that the hardness for general graphs can be obtained directly from \textsc{Independent Set} by adding a universal vertex to the input graph, and further proved that the problem remains $\NP$-complete when restricted to bipartite graphs.

The same idea can be extended to the three parameters considered in this paper. Let $G'$ be obtained from a graph $G$ by adding two adjacent vertices that are universal in $G'$. Then every pair of adjacent vertices of $G'$ forms a hull set. Hence, by the observation made in the Preliminaries section, we have that $\hel(G') = \rad(G') = \rk(G') = \max\{2,\alpha(G')\} = \max\{2,\alpha(G)\}$. Therefore, the $\NP$-hardness and $\W[1]$-hardness of the three parameters on general graphs follows directly from that of \textsc{Independent Set}. We strengthen this observation by showing that all three problems remain hard even when restricted to planar graphs with bounded maximum degree.

\begin{theorem}\label{theo:NPc}
For each of the following three parameters, it is $\NP$-hard to decide
whether a given pair $(G,k)$ satisfies the condition, even when $G$ is a planar graph with maximum degree at most $6$:
\begin{enumerate}
\item[$($i$)$] $\hel(G)\geq k$;
\item[$($ii$)$]$\rad(G)\geq k$;
\item[$($iii$)$]$\rk(G)\geq k$.
\end{enumerate}
\end{theorem}

\begin{proof}
We present a polynomial reduction from \textsc{Independent Set} (IS) on planar graphs with maximum degree at most $3$~\cite{garey1979computers}.

From an instance $(G,k)$ of IS, where $G$ is a planar graph with $\Delta(G)\leq 3$ and $k$ is an integer, we construct an instance $(H,k')$ which will serve for the three desired parameters. Let $\Pi\in \{\text{Helly, Radon, convexly} \}$ denote the corresponding notion of independence.

The graph $H$ arises from $G$ by substituting every edge $e=uv\in E(G)$ by a set of five vertices $W^{e} =\{w_1^{e},\dots,w_5^{e}\}$ and a set of nine edges $\{w_1^{e}w_2^{e},w_2^{e}w_3^{e},w_3^{e}w_4^{e},w_4^{e}w_1^{e},w_2^{e}w_4^{e}\} \cup \{w_4^{e}u, w_4^{e}v, w_5^{e}u, w_5^{e}v\}$.

Given that $|V(G)| = n$ and $|E(G)|=m$ the construction produces a graph $H$ with $|V(H)| = n+ 5m$ and $|E(H)| = 9m$, which is clearly polynomial. Further, the planarity is preserved because each edge $uv$ of a planar embedding of $G$ can be replaced by a planar copy of the subgraph induced by $W^e \cup \{u,v\}$ with $u$ and $v$ on the outer face of the embedding. Moreover, for every vertex $v \in V(G)$, we have $d_H(v) = 2d_G(v) \leq 6$, while, for every $e\in E(G)$, $d_H(w_1^{e}) = d_H(w_3^{e}) = d_H(w_5^{e}) = 2, d_H(w_2^{e})=3, d_H(w_4^{e})=5.$ Hence, $\Delta(H)\leq 6$. 

We show that $G$ has an independent set of order at least $k$ if and only if $H$ has a $\Pi$-independent set of order at least $k'=k+3m.$
Let $I\subseteq V(G)$ be an independent set of $G$ with $|I|\geq k$. We define $S=
I \cup \displaystyle\bigcup_{e\in E(G)} \{w_1^{e},w_3^{e},w_5^{e}\}.$ Notice that $|S|\geq k+3m$. We show in the next claim that every subset of $S$ is a convex set.

\medskip
\noindent \textit{Claim~1.} For every $S'\subseteq S$, $\hullc(S')=S'$.

\smallskip
\noindent \textit{Proof of Claim~1}. Let $S'\subseteq S$ and $e=uv\in E(G)$. In the subgraph of $H$ induced by $S'$, the vertices $w_1^{e}$ and $w_3^{e}$ are isolated, since their neighbors $w_2^{e}, w_4^{e} \notin S$.
Moreover, since $I$ is an independent set of $G$, at most one of $u$ and $v$ belongs to $S'$. Hence, every nontrivial connected component of $H[S']$ is composed by $\{x\}\cup \{w_5^{f} : f\in E(G),\ x\in f,\ w_5^{f}\in S'\}$, for $x\in I \cap S'$.

We show that no vertex in $V(H)\setminus S'$ has two neighbors in the same component of $H[S']$. The possible neighbors of $w_2^{e}$ in $S'$ are $w_1^{e}$ and $w_3^{e}$, which are isolated and hence belong to distinct components. Similarly, the possible neighbors of $w_4^{e}$ in $S'$ are $w_1^{e}$, $w_3^{e}$, $u$, and $v$. The first two are isolated, and at most one of $u$ and $v$ belongs to $S'$, then all belong to distinct components.
Now, let $x\in V(G)\setminus S'$. Its neighbors in $S'$ can only be vertices $w_5^{e}$ corresponding to edges $e$ incident to $x$. By the description of the nontrivial components above, two such vertices $w_5^{e}$ and $w_5^{f}$ could belong to the same component only if the other endpoints of $e$ and $f$ coincide with the same vertex of $I\cap S'$. Since $e$ and $f$ are distinct edges incident to $x$, their other endpoints are distinct. Thus, the neighbors of $x$ in $S'$ belong to distinct components. Therefore, no vertex outside $S'$ has two neighbors in the same component of $H[S']$. Hence, $\hullc(S')=S'$. \hfill $\blacksquare$ 

\medskip

We now show that $S$ is $\Pi$-independent for each of the three parameters. First, for every $x\in S$, Claim~1 gives $\hullc(S\setminus\{x\})=S\setminus\{x\}$, and thus $x\notin\hullc(S\setminus\{x\})$. Hence, $S$ is convexly independent. Moreover, $\displaystyle\bigcap_{x\in S}\hullc(S\setminus\{x\})=\bigcap_{x\in S}(S\setminus\{x\})=\emptyset$, and hence $S$ is Helly-independent. Finally, consider any partition $S=S_1 \cup S_2$. By Claim~1, $\hullc(S_1)=S_1$ and $\hullc(S_2)=S_2$. Therefore, $\hullc(S_1)\cap\hullc(S_2)=\emptyset$, and hence $S$ is Radon-independent. Thus, $H$ has a $\Pi$-independent set of order at least $k+3m$.

\smallskip
Conversely, let $S$ be a $\Pi$-independent set of $H$ with $|S|\geq k+3m$. In all three cases, $S$ is convexly independent. Therefore, it suffices to consider convex independence in the remainder of the proof. The following claim bounds the number of vertices of $S$ that can be chosen from each gadget. Let us fix $S_G=S\cap V(G)$.

\medskip
\noindent \textit{Claim~2.} For every edge $e=uv\in E(G)$, it holds that:
\begin{enumerate}[(\roman*)]
\item if $|S_G \cap \{u,v\}|=2$, then $|S\cap W^{e}|\leq 2$;
\item if $|S_G \cap \{u,v\}|\leq 1$, then $|S\cap W^{e}|\leq 3$.
\end{enumerate}

\smallskip
\noindent \textit{Proof of Claim~2}. First, since $S$ is convexly independent, $S\cap(W^{e}\cup \{u,v\})$ is also convexly independent in  $H[W^{e}\cup\{u,v\}]$. In addition, since every pair of adjacent vertices in $\{w_1^{e},\dots,w_4^{e}\}$ forms a hull set of $H[\{w_1^{e},\dots,w_4^{e}\}]$, we get $|S\cap \{w_1^{e},\dots,w_4^{e}\}| \leq 2$ (I).

\begin{enumerate}[(\roman*)]
\item Suppose that $|S_G\cap \{u,v\}|=2$, that is, $u,v\in S_G$. Assume, by contradiction, that $|S\cap W^{e}|\geq3$. By (I) we have $w_5^{e}\in S$, then $w_4^e$ forms a cycle with $u,w_5,v$ and $w_4^e \in \hullc (\{u,w_5,v\})$. Thus, if $w_4^{e}\in S$, then $w_4^{e}\in\hullc(S\setminus\{w_4^{e}\})$, a contradiction. Hence, $w_4^{e}\notin S$, and consequently $|S\cap \{w_1^{e},w_2^{e},w_3^{e}\}|\geq2$. Let $x,y \in S\cap \{w_1^{e},w_2^{e},w_3^{e}\}$, with $x\neq y$. Since $w_4^{e} \in \hullc (\{u, w_5^{e},v\})$ and $\{x,w_4^{e}\}$ is a hull set of $H[\{w_1^{e},\dots,w_4^{e}\}]$, we have that $y\in\hullc(S\setminus\{y\}),$ a contradiction. 

\item Suppose that $|S_G\cap\{u,v\}| \leq 1$. If $|S\cap W^{e}| \geq 4$, then at least three vertices of $\{w_1^{e},\dots,w_4^{e}\}$ belong to $S$, contradicting (I). \hfill $\blacksquare$ 
\end{enumerate}

Next, we distinguish the edges of the induced subgraph $G[S_G]$ to establish an upper bound of the number of vertices in~$S$.
By Claim~2, every edge $e \in  E(G[S_G])$ contributes at most two vertices of $W^{e}$ to $S$, while every edge $e \notin E(G[S_G])$ contributes at most three. This implies that $|S\setminus S_G|\leq 2|E(G[S_G])| + 3\bigl(m-|E(G[S_G])| \bigr) =3m-|E(G[S_G])|,$ and therefore $|S|\leq 3m+|S_G|-|E(G[S_G])|.$

By choosing one endpoint of each edge of $G[S_G]$, we obtain a vertex cover of $G[S_G]$ of size at most $|E(G[S_G])|$. Consequently, $G[S_G]$ has an independent set of order at least $|S_G|-|E(G[S_G])|$, and thus $
\alpha(G)\geq |S_G|-|E(G[S_G])|.$
Since $|S|\geq k+3m$, we conclude that
$$
k+3m\leq |S|\leq 3m+|S_G|-|E(G[S_G])|\leq 3m+\alpha(G),
$$ and hence $\alpha(G)\geq k$. Therefore, $G$ has an independent set of order at least $k$. \hfill \qed
\end{proof}

It is worth mentioning that the $\W[1]$-hardness of the rank parameterized by the solution size, known for general graphs, does not extend to planar graphs. Indeed, every independent set is simultaneously Helly-, Radon-, and convexly independent in the cycle convexity. By the Four Color Theorem, every planar graph $G$ has an independent set of order at least $|V(G)|/4$. Therefore, for any $\Pi\in\{\text{Helly, Radon, convexly}\}$, if $|V(G)|\geq 4k$, then $G$ necessarily contains a $\Pi$-independent set of order at least $k$. Otherwise, $|V(G)|<4k$, and the corresponding parameter can be computed by exhaustive enumeration in time depending only on $k$. Consequently, all three problems are fixed-parameter tractable on planar graphs when parameterized by $k$. Thus, while the result of Theorem~\ref{theo:NPc} shows that the problems remain $\NP$-hard even for planar graphs of maximum degree at most $6$, this hardness cannot be strengthened to $\W[1]$-hardness with respect to the solution size, unless $\FPT=\W[1]$.

\section{Extremal values}

If $\pi(G)$ denotes any one of $\hel(G)$, $\rad(G)$, or $\rk(G)$, then $\pi(G)=n(G)$ if and only if $G$ is a forest. Motivated by this characterization, we investigate the next two extremal values. In particular, we characterize the graphs for which the Helly number, Radon number, and rank attain the values $n(G)-1$ and $n(G)-2$. Interestingly, these three parameters coincide in both cases. We begin with the case $\pi(G)=n(G)-1$.

\begin{theorem}\label{unicyclic}

 Let $G$ be a graph of order $n\geq 3$. Then the following are equivalent. 
 
 \begin{enumerate}
 \item[$(1)$] $\rk(G)=n-1$.
 \item [$(2)$]$\hel(G)=n-1$.
 \item[$(3)$] $\rad(G)=n-1$.
     \item[$(4)$] $G$ is a unicyclic graph.
 \end{enumerate}
\end{theorem}
\begin{proof}

$(1)\iff (4)$

Suppose that $\rk(G)=n-1$. Then there exists a convexly independent set $S=V(G)\setminus\{u\}$ for some $u\in V(G)$. Since every forest has rank $n$, the graph $G$ contains a cycle. Suppose, to the contrary, that $G$ contains two distinct cycles, say $C_1$ and $C_2$. Since $G[S]$ is a forest, we have $u\in V(C_1)\cap V(C_2)$. Choose $v\in V(C_1)\setminus V(C_2)$. Then $v\in\intv(S\setminus\{v\})\subseteq\hullc(S\setminus\{v\})$, contradicting the convex independence of $S$. Therefore, $G$ contains exactly one cycle.

Conversely, suppose that $G$ has a unique cycle $C_m$, where $3\leq m\leq n$, and let $u\in V(C_m)$. Fix $S=V(G)\setminus\{u\}$. Then for any $v\in S$, $
v\notin\hullc(S\setminus\{v\})$, and so $S$ is convexly independent. Thus $\rk(G)\geq n-1$. On the other hand, since $G$ is not a forest, $\rk(G)\leq n-1$, and consequently $\rk(G)=n-1$.

 $(2) \iff (4)$

Assume that $\hel(G)=n-1$. Then, for some $u\in V(G)$, the set $S=V(G)\setminus\{u\}$ is $H-$independent in $G$. Since trees are characterized by helly number $n$, it follows that $G$ contains a cycle. Now, assume to the contrary that $G$ has at least two cycles, say $C_1$ and $C_2$. If $u\notin V(C_1)\cup V(C_2)$, then for every $v\in V(C_1)\cup V(C_2)$, we get $v\in \displaystyle\bigcap_{v\in S}\hullc (S\setminus\{v\})$, $S$ is $H$-dependent. Suppose that $u$ belong to exactly one cycle, say, $u\in V(C_1)\setminus V(C_2)$. Then, for every $v\in V(C_2)$, $v\in \displaystyle\bigcap\hullc (S\setminus\{a\})$ for each $a\in S$, again $S$ is $H$-dependent. Finally, consider $u\in V(C_1)\cap V(C_2)$. Then $\displaystyle u\in \bigcup_{v\in S}\hullc (S\setminus \{v\})$ and hence $S$ become $H$-dependent. Therefore, $G$ contains precisely one cycle and is therefore unicyclic.

Let $G$ be a unicyclic graph of order $n\geq 3$ with a cycle $C_m$, where $3\leq m \leq n$ and let $u\in V(C_m)$. Consider the set $S=V(G)\setminus \{u\}$. Then, for any $v\notin S$, $\hullc (S\setminus \{v\})=S\setminus\{v\}$ and hence it is $H$-independent in $G$ and thus, $\hel(G)\geq n-1$.

Using similar arguments as above, we can prove that $(3)\iff (4)$.
 
\qed
\end{proof}

Next, we turn our attention to characterizing the class of graphs attaining the extremal value $n(G)-2$. To this end, we introduce the following three graph classes.
 \begin{itemize}
    \item Let $\mathscr{F}_1$ denote the family of all bicyclic graphs.
  \item Let $\mathscr{F}_2$ denote the family of graphs obtained from a tree $T$ by adding two new adjacent vertices $x$ and $y$ and edges between $\{x,y\}$ and $V(T)$ such that $N(x)=N(y)=\{u,v\}$ for some $u,v\in V(T)$. 
\item Let $\mathscr{F}_3$ be the family of graphs obtained from $k\geq2$ pairwise vertex-disjoint trees $T_1,T_2,\ldots,T_k$ by adding two new vertices $x$ and $y$, which may be adjacent or non-adjacent, and adding edges between $\{x,y\}$ and $\displaystyle\bigcup_{i=1}^k V(T_i)$ such that the resulting graph is connected and each of $x$ and $y$ has at most one neighbour in each tree $T_j$, $1\leq j\leq k$. With this notation, we have the following result.
    \end{itemize}

\begin{theorem}
     Let $G$ be a graph of order $n\geq 4$. Then the following are equivalent.
     \begin{enumerate}
     \item [$(1)$] $\rk(G)=n-2$.
         \item [$(2)$]$\hel(G)=n-2$.
         \item [$(3)$]$\rad(G)=n-2$.
         \item [$(4)$]$G\in \displaystyle\bigcup_{i=1}^3 \mathscr{F}_i$.
     \end{enumerate}
      \end{theorem}
\begin{proof}

$(1)\implies (4)$

Let $G$ be a graph of order $n$ with $\rk(G)=n-2$, and let $S$ be a maximum convexly independent set such that $V(G)\setminus S=\{x,y\}$. Let $T_1,T_2,\ldots,T_k$, $k\geq1$, denote the components of the induced subgraph $G[S]$. Since every convexly independent set induces a forest, each $T_i$ is a tree. We consider the following two cases.\\
\textbf{Case 1:} $xy\notin E(G)$.
First, suppose that $k=1$. Then $G[S]$ is a tree. Since $G$ is connected and $S$ is a maximum convexly independent set, it follows that both $x$ and $y$ have at least two neighbours in $S$. Suppose that $x$ has three distinct neighbours $x_1,x_2,x_3$ in $S$. Then $x_1\in\intv^2(S\setminus\{x_1\})\subseteq\hullc(S\setminus\{x_1\})$, contradicting the convex independence of $S$. Thus, $x$ has exactly two neighbours in $T_1$. Similarly, $y$ has exactly two neighbours in $T_1$. Consequently, $G$ is bicyclic, and hence $G\in\mathscr{F}_1$.

Next, assume that $k\geq2$. $G$ contains at least two cycles. Let $C_1,C_2,\ldots,C_\ell$ be the distinct cycles of $G$. Since every convexly independent set induces a forest, it is clear that every cycle of $G$ must contain at least one of $x$ and $y$.

Suppose first that $V(C_i)\cap\{x,y\}=\{x\}$ for some $i$. It follows that $V(C_j)\cap\{x,y\}\neq \{x\}$ for every $j\neq i$; otherwise, $S$ would be convexly dependent. Without loss of generality, let $i=1$ and assume that $V(C_1)\setminus\{x\}\subseteq V(T_1)$. As in the preceding case, each of $x$ or $y$ has at most one neighbour in each component $T_i$, $1\leq i\leq k$. Consequently, $V(C_j)\cap V(T_1)=\emptyset$ for every $j\geq2$. Now, suppose that there exists a cycle, say $C_2$, such that $x\in V(C_2)$. Then $V(C_2)\cap\{x,y\}=\{x,y\}$. In this case, $G$ must be bicyclic and so $G\in\mathscr{F}_1$. Otherwise, let $C_3$ be another cycle and let $u$ be a neighbour of $y$ on $C_3$. Then $u\in\intv^2(S\setminus\{u\})\subseteq\hullc(S\setminus\{u\})$, contradicting the convex independence of $S$. Hence we can assume that $V(C_j)\cap\{x,y\}=\{y\}$ for every $j\geq2$. By applying the preceding argument with the roles of $x$ and $y$ interchanged, we conclude that there is exactly one cycle containing $y$. Thus, in this case, $G$ also has exactly two cycles and so $G\in\mathscr{F}_1$.

Henceforth, we may assume that every cycle $C_i$ of $G$ contains both $x$ and $y$. Consequently, each of $x$ and $y$ has at most one neighbour in every component $T_j$, $1\leq j\leq k$. Together with the connectivity of $G$, these observations imply that $G\in\mathscr{F}_3$.\\
\textbf{Case 2:} $xy\in E(G)$. As in Case 1, we first consider $k=1$. By analogous arguments, both $x$ and $y$ have exactly two neighbours in $T_1$. If $|N(x)\cup N(y)|\geq3$, then, for any $u\in N(x)$, we have $u\in \hullc(S\setminus\{u\})$, contradicting the convex independence of $S$. Hence, $N(x)=N(y)$. Consequently, $G\in\mathscr{F}_2$. The rest of the proof is also analogous to that of Case 1 and is therefore omitted.

It is straightforward to verify that every graph belongs to the family $\displaystyle\bigcup_{i=1}^3\mathscr{F}_i$ satisfies $\rk(G)=n(G)-2$. 

$(2)\iff (4)$

Let $G$ be a graph with $\hel(G)=n-2$. Then there exists an $H$-independent set $S\subseteq V(G)$ with $|S|=n-2$. Let $S=V(G)\setminus\{x,y\}$ for some $x,y\in V(G)$. Since every $H$-independent set induces a forest, $S$ induces a forest in $G$. Let $T_1,T_2,\ldots,T_k$ be the trees of this forest.

\textbf{Case 1:} $xy\notin E(G)$.

First assume $G[S]$ is a tree. If $x$ or $y$ has more than two neighbours in $S$, then $S$ is $H$-dependent. Suppose that $x$ has two neighbours in $S$ and $y$ has exactly one neighbour in $S$. Then $G$ is a unicyclic graph, and hence, by Theorem~\ref{unicyclic}, $\hel(G)=n-1$. Similarly, if $y$ has two neighbours in $S$ and $x$ has exactly one neighbour in $S$, then $\hel(G)=n-1$. So $S$ is Helly independent only when $x$ and $y$ has exactly two neighbours in $G[S]$. Therefore $G\in \mathscr{F}_1$.

Let $G[S]$ is a forest with components $T_1,T_2,\ldots,T_k$. As in the previous case, neither $x$ nor $y$ can have more than two neighbours in two or more distinct components of $S$. Then, Theorem \ref{unicyclic} implies that $G$ contains at least two cycles, say $C_1,C_2,\ldots,C_l$. It follows that every cycle of $G$ must include at least one of the vertices $x$ and $y$. First, assume that $\{x\}=V(C_i)\cap \{x,y\}$ for some  $i\in \{1,2,3,\ldots,l\}$, while $\{y\}= V(C_j)\cap \{x,y\}$ for some $j\neq i$. Then, for any vertex $u\in S, \hullc(S\setminus\{u\})=S\setminus\{u\}$ and is Helly independent in $G$ having exactly two cycles and in this case $G\in \mathscr{F}_1$.

Now, suppose that $G$ contains a cycle $C_i$ for which $\{x\}= V(C_i)$ and another cycle $C_j$ with $V(C_j)\cap \{x,y\}=\{x,y\}$. In this case also $S$ is $H$-independent and $G$ is bicyclic. On the other hand beside $C_i$ and $C_j$ there is a cycle $C_k$ with $y$ has a neighbour, say $u$ in $C_k$. Then, it follows that $u\in \displaystyle\bigcap_{a\in S}\hullc(S\setminus\{a\})$, a contradiction to the Helly independence of $S$. Therefore, $G$ is bicyclic and hence $G\in \mathscr{F}_1$. Hereafter, we may assume that both $x$ and $y$ belong to every cycle of $G$. Then, for each $u\in S, \hullc(S\setminus\{u\})=S\setminus\{u\}$ and hence $S$ is $H$-independent. These observations along with the connectivity of $G$ imply that $G\in \mathscr{F}_3$.\\

\textbf{Case 2:} $xy\in E(G)$. If $G[S]$ is a tree, by the same arguments in Case 1, each of $x$ and $y$ has exactly two neighbours in $G[S]$. If $|N(x)\cup N(y)|\geq 3,$ then for some  $u\in N(x)$, $u\in \displaystyle\bigcap_{a\in S}\hullc(S\setminus\{a\})$, implies that $S$ is $H$-dependent, a contradiction. Therefore, $N(x)=N(y)$ and consequently, $G\in \mathscr{F}_2$. Since the remaining arguments are similar to that of Case 1, we omit the proof.\\
$(1)\implies (4)$ is straight forward and hence we exclude the proof. 

By using the similar proof technique we get  $(3)\iff (4)$.
\qed
\end{proof}

\begin{figure}[ht]
\centering

\begin{tikzpicture}[
    vertex/.style={
        circle,
        draw=black,
        fill=black,
        inner sep=1.6pt
    },
    lab/.style={
        font=\normalsize
    }
]


\node[vertex] (a1) at (0,1.6) {};
\node[vertex] (a2) at (0,0) {};
\node[vertex] (a3) at (2.1,0) {};
\node[vertex] (a4) at (2.1,1.6) {};

\node[vertex] (a5) at (4.0,0.8) {};

\node[vertex] (a6) at (5.9,1.6) {};
\node[vertex] (a7) at (5.9,0) {};
\node[vertex] (a8) at (8.0,0) {};
\node[vertex] (a9) at (8.0,1.6) {};


\node[lab, above left] at (a1) {$a_1$};
\node[lab, below left] at (a2) {$a_2$};
\node[lab, below] at (a3) {$a_3$};
\node[lab, above] at (a4) {$a_4$};

\node[lab, below, yshift=-4pt] at (a5) {$a_5$};

\node[lab, above left] at (a6) {$a_6$};
\node[lab, below left] at (a7) {$a_7$};
\node[lab, below] at (a8) {$a_8$};
\node[lab, above] at (a9) {$a_9$};


\draw (a1)--(a2);
\draw (a2)--(a3);
\draw (a3)--(a4);
\draw (a4)--(a1);
\draw (a1)--(a3);
\draw (a2)--(a4);


\draw (a3)--(a5);
\draw (a4)--(a5);
\draw (a5)--(a6);
\draw (a5)--(a7);


\draw (a6)--(a7);
\draw (a7)--(a8);
\draw (a8)--(a9);
\draw (a9)--(a6);
\draw (a6)--(a8);
\draw (a7)--(a9);

\end{tikzpicture}

\caption{A graph satisfying $\hel(G)=\rad(G)\neq\rk(G)$.}
\label{fig:example}

\end{figure}

\section{Conclusion}

Here we studied the Helly number, Radon number, and rank of graphs with respect to cycle convexity, obtaining structural characterizations of their extremal and near-extremal values. In particular, we characterized the graphs attaining the values $n-1$ and $n-2$, and, interestingly, the three parameters coincide in both cases. Moreover, we observe that if $\rk(G)\leq 3$, then $\hel(G)=\rad(G)=\rk(G)$. Since the proof is straightforward, we omit the details. However, this property does not extend to larger values: The graph in Figure~\ref{fig:example} satisfies $\hel(G)=\rad(G)=3<4=\rk(G)$. The set $S=\{a_1,a_2,a_8,a_9\}$ is convexly independent. However, since $a_5\in\displaystyle\bigcap_{a_i\in S}\hullc(S\setminus\{a_i\})$, the set $S$ is Helly dependent. Similarly, $S$ is Radon dependent. On the other hand, every four-element set of vertices distinct from $S$ is both Helly dependent and Radon dependent, by the fact that the subgraph induced by any independent set is a forest, as established above. Thus, \(4\) is the smallest value at which the three parameters can differ. Our observations also suggest the following conjecture.

\begin{conjecture}
For every connected graph \(G\) of order \(n\geq 2\), we have
$\hel(G)=\rad(G)$.
\end{conjecture}
 Determining whether the Helly number and Radon number always coincide, and understanding the possible gap between these parameters and the rank, are natural directions for future research.
 
From the computational perspective, we established hardness results for the corresponding threshold decision problems, which remain hard even for planar graphs of bounded maximum degree. These results enhance our understanding of the structural and computational aspects of cycle convexity and provide directions for further study.

 \subsubsection*{Acknowledgment:} Revathy S.Nair acknowledges the financial support from the University of Kerala, for providing the University Junior Research Fellowship (AcEVI 3217/2025/UOK dated 10/04/2025).\\

\bibliographystyle{amsplain}
\bibliography{cycle}

\end{document}